%% file: Indexhelpers_draft3.tex
\documentclass[11pt,final,journal,letterpaper,onecolumn]{IEEEtran}
\usepackage{cite}
\usepackage{graphicx,color,epsfig,rotating}
\usepackage{amsfonts,amsmath,amssymb}
\usepackage{subfig}

\input{macros}

\newtheorem{definition}{Definition}
\newtheorem{theorem}{Theorem}
\newtheorem{lemma}{Lemma}
\newtheorem{corollary}{Corollary}

\newtheorem{proof}{Proof}

\newtheorem{problem}{Problem}

\begin{document}

\title{Index Coding Problem with Side Information Repositories}
\author{
\IEEEauthorblockN{Karthikeyan Shanmugam and  Alexandros G. Dimakis } 
\IEEEauthorblockA{ \\Department of Electrical and Computer Engineering \\
University of Texas at Austin \\
Austin, TX 78712-1684 \\
karthiksh@utexas.edu,dimakis@austin.utexas.edu \\}
\and
\IEEEauthorblockN{Giuseppe Caire }
\IEEEauthorblockA{\\Department of Electrical Engineering \\
University of Southern California\\
Los Angeles, CA- 90089-2560 \\
caire@usc.edu}
}

\date{\today}

\maketitle

{\abstract
               To tackle the expected enormous increase in mobile video traffic in cellular networks, an architecture involving a base station along with caching femto stations (referred to as helpers), storing popular files near users, has been proposed \cite{golrezaei2011femtocaching}. The primary benefit of caching is the enormous increase in downloading rate when a popular file is available at helpers near a user requesting that file. In this work, we explore a secondary benefit of caching in this architecture through the lens of \textit{index coding}. We assume a system with $n$ users and constant number of caching helpers. \textit{Only} helpers store files, i.e. have side information.  We investigate the following scenario: Each user requests a distinct file that is \textit{not} found in the set of helpers nearby. Users are served coded packets (through an \textit{index code}) by an omniscient base station. Every user decodes its desired packet from the coded packets and the side information packets from helpers nearby. We assume that users can obtain any file stored in their neighboring helpers without incurring transmission costs.               
                With respect to the index code employed, we investigate two achievable schemes: 1) \textit{XOR coloring} based on coloring of the side information graph associated with the problem and 2)Vector XOR coloring based on fractional coloring of the side information graph. We show that the general problem reduces to a canonical problem where every user is connected to exactly one helper under some topological constraints.  For the canonical problem, with constant number of helpers ($k$), we show that the complexity of computing the best XOR/vector XOR coloring schemes are polynomial in the number of users $n$. The result exploits a special \textit{complete bi-partite structure} that the side information graphs exhibit for any finite $k$. 
}

\section{Introduction}
	           The main challenge that faces today's operators of 3GPP LTE Advanced enabled cellular networks is an enormous increase in mobile video traffic. The mobile video traffic in a cell is also expected to have a  lot of redundancy in demand for stored videos (Youtube-like short videos) due to few videos being more popular compared to others \cite{zink2009characteristics}. To tackle the traffic load bottleneck by exploiting this redundancy in requests, an architecture involving small WiFi enabled \textit{caching helpers} was proposed in \cite{golrezaei2011femtocaching}. In this framework, few femto-like stations would be distributed in a cell, controlled by the base station, with an additional feature of caching popular files near users. The motivation for caching at the helpers is to reduce the bandwidth demand on the backhaul to these helpers. A high bandwidth backhaul network is expensive to deploy. If a user request has a high cache hit in one of the helpers nearby, then the user downloads the file from the helper through a high speed WiFi link without having to rely on the base station or the backhaul.  This is the primary benefit of caching.

		            There is a potential hidden benefit to be leveraged due to such caching in helpers even when there are cache misses for all users. The benefit is gained when base station sends few 'coded' packets where coding is done over the set of packets requested by all the users. A user terminal decodes its desired packet from the coded transmissions by downloading some \textit{side information} available as cached data from a subset of the helpers it is connected to even when those helpers do not have the user's desired packet. The total number of coded transmissions by the base station may be lower compared to the the \textit{naive scheme} where the base station transmits all user requests in sequence. This is because if some helpers in the network cache previous requests in the network and if requests arising from the entire network have redundancy over time then the following scenario may occur: user $A$ and user $B$ both have cache misses. But user $A$ has a cache hit with the helper near user $B$ but located far away from user $A$ and vice versa. Now the packets intended for both users can be XORed saving one transmission for the base station. The user-helper links are high speed links and hence pulling side information required to decode is assumed to involve zero cost.  In this work, we assume that a connectivity graph between helpers and users, an arbitrary cache state of helpers (users do not cache) and a set of simultaneous \textit{distinct} user requests, such that every user has a \textit{cache miss}, are given as the input. We investigate the problem of minimizing the number of broadcast transmissions by the base station in such a wireless network which is a special case of the \textit{index coding} problem.

	            The general index coding problem is a noiseless broadcast problem that was introduced in \cite{bar2011index}. There are $n$ users, each wants a distinct packet and has multiple \textit{side information} packets, i.e. some packets desired by others. The problem of multiple users requesting the same packet is more general and it requires representation using hypergraphs or another special  representation in terms of  bipartite graphs \cite{neely2011dynamic}\cite{alon2008broadcasting}. The problem instance when each user requests a different packet can be represented using a directed \textit{side information graph} where every vertex represents a user (and equivalently a packet desired by the user as every request is distinct). There is a directed edge from user $i$ to user $j$ if user $i$ has the packet desired by $j$ as side information. In the very special case when side information is symmetric, i.e. user $i$ has user $j$'s packet as side information whenever user $j$ has user $i$'s packet as side information, the problem instance can be represented by an undirected side information graph. 
	            	
		 Each packet is drawn from an alphabet of size $t$ bits. An index code is a function that maps $n$ packets, each of $p$ bits, to a packet of size $t$ bits. The encoded message of $t$ bits is transmitted over a noiseless broadcast channel. An index code is valid if every user can use the encoded message and the side information it has to decode the desirable packet (of $p$ bits), i.e. there is a suitable decoding function. The index coding problem is to compute the minimum broadcast rate $\frac{t}{p}$ over all possible index codes. The minimum broadcast rate is denoted by $\beta$. The general problem of computing minimum broadcast rate over all possible encoding functions is very hard to characterize \cite{blasiak2010index}.  
	           
	         Let the message alphabet, represented by $p$ bits,  be a finite field. Let $t$ be a multiple of $p$, i.e. encoded message is a sequence of elements drawn from the field. If every subsequent field element in the encoded message is a linear function of the $n$ desired messages ($n$ field elements), then such an index code is a scalar linear index code over that finite field. 
	
		            \textbf{Prior Work:} It was shown \cite{bar2011index} that the length of the optimum scalar linear index code (over a given finite field) is equal to the graph parameter minrank, over that field, of the directed side information graph. This is the directed graph analog of the graph parameter introduced in \cite{haemers1979some} to upper bound the Shannon capacity of an undirected graph. Computing minrank is known to be NP-hard \cite{peeters1996orthogonal} \cite{dau2012optimal}. In \cite{lubetzky2009nonlinear}, the authors show the existence of graphs where the multiplicative gap between binary minrank and the optimum broadcast rate $\beta$ grows polynomially in the number of nodes $n$. Some recent works  \cite{ong2012optimal} \cite{berliner2011index} also dealt with tractability of computing the binary minrank for restricted classes of graphs.  In \cite{haviv2011linear},  the minrank parameter of Erd\H{o}s-Renyi graphs $G(n,p)$ were analyzed. Authors of \cite{chlamtavc2012linear} provide a polynomial time algorithm based on graph coloring for approximating minranks for undirected graph instances whose minrank is $k$. Recently, interference alignment approach has been applied to the problem of index coding \cite{maleki2012index}.   In \cite{blasiak2010index}, a set of linear programs were developed whose optima effectively provide lower and upper bounds on the optimum broadcast rate $\beta$.  Fractional coloring and coloring of the complement of the undirected side information graph are the most well known popular upper bounds to the optimum rate $\beta$ \cite{blasiak2010index} yielding achievable linear binary schemes. Apart from the index coding literature, a slightly different approach combining the cache design and coding problem was taken in \cite{maddah2012fundamental} where all users cache data (no helpers). The problem is to design caches first and then the coding scheme to minimize packets transmitted for the worst case user requests all of which are known to come from a fixed library of files with possible overlaps among user requests.
		             		            		             		          		
	             In this work, we study a special case of the index coding problem with distinct requests motivated by the caching architecture mentioned above in a wireless network setting. Here, a side information packet is a packet cached only in the helpers and a user has access to them at zero cost if it is connected to the helper containing it. The base stations transmit coded packets to satisfy distinct requests of $n$ users.  We study the algorithmic tractability for the well known linear achievable schemes, based on graph coloring, when user side information is restricted to come from a constant number of caching helpers. We note that intractability of coloring based achievable schemes for the general index coding problem arise due to the freedom available in choosing an arbitrary side information set for each user as part of the problem instance.
	             	            	           
	            A typical scenario of the problem of interest in this work is illustrated in Fig. \ref{Fig:IndexCodinghelper}. In the figure, users request different packets, none of which are found in their respective set of neighboring helpers. In the figure, user $1$ downloads packet $P5$ from helper $1$ in order to decode its desired packet $P1$ from the coded packet $P1 + P5$ transmitted by the base station. In this example, $7$ user requests can be satisfied by $4$ coded transmissions from the base station, plus some local transmission from the helpers, which can take place at a much faster speed, due to the high capacity of short range links.  Hence, the presence of helpers with side information can alleviate the base station congestion even though the requests are not found in the helper caches, as in this example.  Helpers have an arbitrary cache state given as part of the input. Further, the connectivity graph between users and helpers in specified. The base station transmits coded packets. 
	            
	            We investigate the computational complexity of a coding scheme where every user needs to listen to just one coded packet that contains the desired packet XORed with other packets. At the time of decoding, the user needs  to just download an extra packet from neighboring helpers, which is an appropriate 'mix' of side information packets present at helpers nearby to cancel the ``interference'' formed by the undesired packets, and decode the desired packet.  We call this coding scheme an \textbf{XOR coloring} scheme. XOR coloring can be shown to be based on coloring the complement of the undirected side information graph. The XOR coloring scheme is a scalar scheme, i.e. every transmission is linear in the packets of different users.  We also investigate a vector coding scheme called vector XOR coloring that generalizes the above scheme and is based on fractional coloring of the complement of the side information graph. Every user request consists of a set of sub-packets. Under this scheme, the base station codes over these sub-packets in such a way that every user needs to listen to exactly one coded sub-packet to recover one of the desired sub-packets using side-information from the helpers. 
	            	                   

	\begin{figure}
	 \centering
	  \includegraphics[width=8cm]{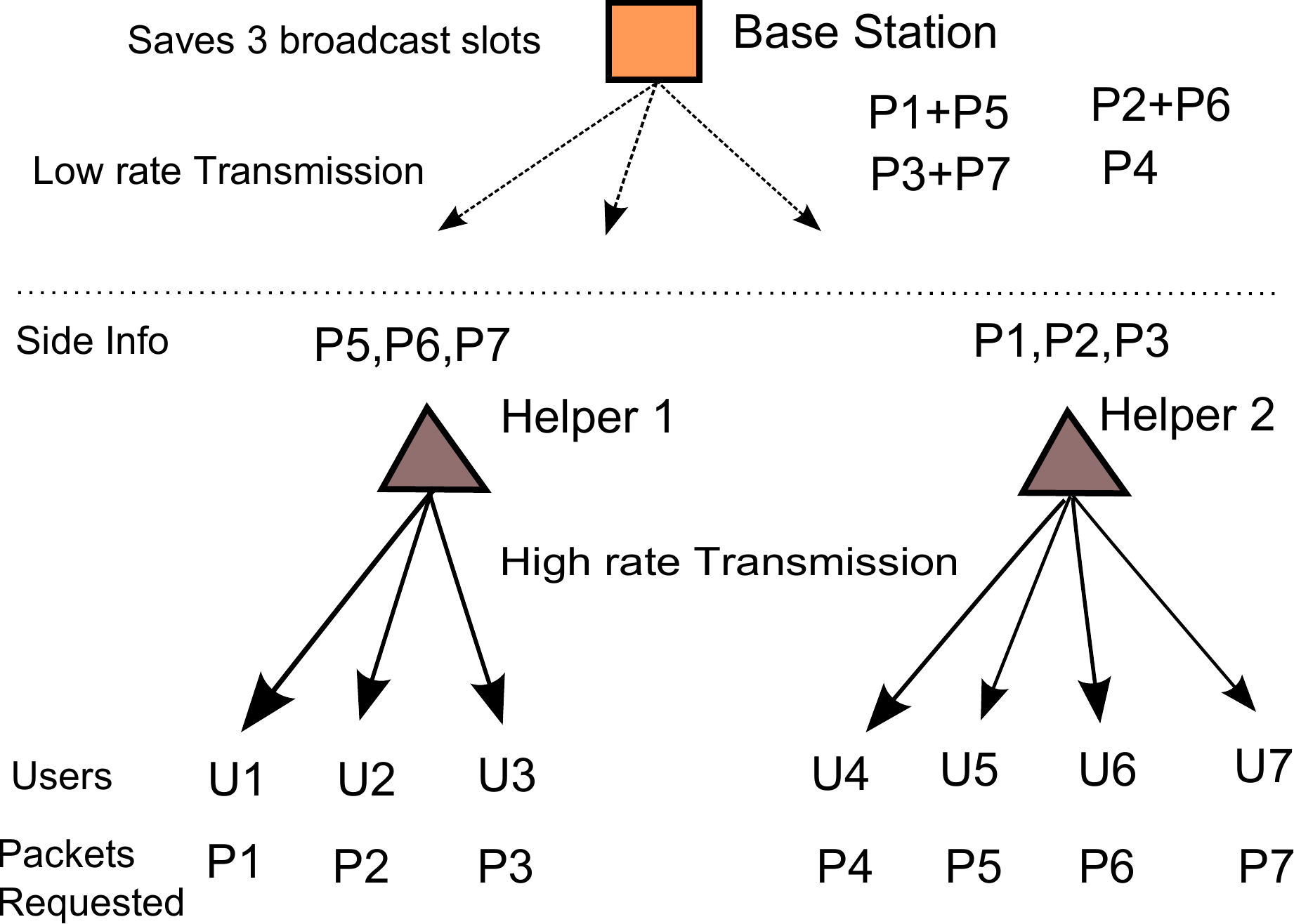}
	  \caption{Index Coding Problem with Caching helpers}
	  \label{Fig:IndexCodinghelper}
	 \end{figure}

	  \textbf{Our Contributions:} We treat the special case of the index coding problem in a wireless network, where the side information is available at a few caching helpers that serve nearby users through short range WiFi links (with no cost) and all users can receive data from the base station. We assume a network with one base station (source), $k$ (a constant) helpers and $n$ users.  With interference constraints on helper placement and planar topology, one can reduce a general problem, where users are connected to multiple helpers, to a canonical form in which users are connected to just one helper, while the number of helpers in the resulting `canonical network' is increased by at most a constant factor with respect to the original network. For the equivalent problem reduced to such a ``canonical form'', the problem of finding the optimal XOR coloring scheme, when $k$ is a constant that does not scale with $n$, reduces to a problem of multi coloring on a fixed prototypical graph uniquely determined by $k$ only. The integer programming formulation of such a multi coloring problem yields a optimal XOR coloring scheme computable in time polynomial in the number of users $n$. We further show that vector XOR coloring reduces to fractional multi coloring on the prototypical graph whose computational complexity is polynomial in $n$ when $k$ is a constant ($ 2^{O(k^2)}n$).  For $k = 2$, the XOR coloring code is the optimal linear scalar index code. For $k=3$, XOR coloring is also shown to be the optimal binary scalar linear index code when considering the undirected side information graph obtained by throwing away uni-directed edges. 

	
	\section{Problem Definition}\label{probdefsec}

	    We consider a special instance of the index coding problem involving constant number of helpers holding cached packets that act as side information, 
	    motivated by a wireless setting with a central broadcasting agent, i.e. a base station, serving all users in a cellular network. 
	    Access to side information is provided by helpers with storage, holding previously cached content, and equipped with another wireless 
	    interface to serve the user requests \textit{only} using their cached content. 
	    Every helper has a limited connectivity range and, when a user is within this range, we assume that the helper user communication incurs no cost due to high downloading rates from the helpers. In this sense, side information \textit{available} to a user includes all the cached content of a helper if the user is in the connectivity range of the helper. When we deal with binary scalar coding schemes in this work, packet size is taken to be $1$ bit ($p$=1) without loss of generality since the same scalar scheme is repeated over all the bits of a packet. Formally,  the problem instance can be stated as follows:
	    
	\begin{problem}
	\label{HelperProblem}
	(\textit{Index Coding with Helpers: ICH$(\{C_i,S_i \})$} The problem instance is given by $n$ users, indexed by elements from set ${\cal N}=\{1,2 \ldots n\}$, served by the base station, requesting $n$ different packets $x_1,x_2 \ldots x_n \in \mathbb{F}_2$ respectively. There are $k$ helper nodes denoted by ${\cal H}=\{H_1,H_2\ldots H_k\}$,such that helper $H_i$ contains a subset $S_i \subseteq \{x_1,\ldots, x_n\}$ of the $n$ packets. User requests are not found in its neighboring helpers (\textit{cache misses}). Letting $C_j \subseteq {\cal N}$ denote the set of users connected to helper $H_j$, cache miss condition implies $x_i \notin S_j$ for all $i \in C_j$.  
       \end{problem}
		 	
	\begin{definition}\label{SideInfograph}
	   A \textit{Side Information graph} is a directed graph $G_d(V,E_d)$ where each vertex $i$ corresponds to a user $i$ and there is a directed edge from $i$ to $j$ denoted by the ordered pair $(i,j)$ when user $j$ has packet $x_i$ as side information. This notion holds only for the case when user $i$ requests a single packet $x_i$ and all user requests are for different packets. \hfill $\lozenge$
	\end{definition}
	  Any index coding problem with distinct requests has a directed side information graph characterizing it. The ICDH problem when $k=n$ and $|C_i|=1, ~ \forall i$ is the general index coding problem instance on a directed side information graph.
	
	\begin{definition}\label{UnSideInfograph}
	  An \textit{underlying undirected side information graph}  $G(V,E)$ of a given directed side information graph $G_d$ is obtained by replacing every two directed edges $(i,j)$ and $(j,i)$ by an undirected edge, denoted $\{i,j\}$, and throwing away any directed edge $(i,j)$ that does not have an oppositely directed counterpart, namely $(j,i)$. \hfill $\lozenge$
	\end{definition}
	
We introduce some standard notation for a few graph parameters. Given an undirected graph $G$, $\bar{G}$ denote its complement. $\alpha(G),\omega(G),\bar{\chi}(G),\chi(G)$ are the independence number, clique number, clique cover number and the chromatic number, respectively of $G$. Since an independent set in $G$ is the clique in $\bar{G}$, we have $\chi(\bar{G})= \bar{\chi}(G)$ and $\alpha(G) = \omega(\bar{G})$. In a feasible coloring of an undirected graph, neighboring nodes receive different colors. 
	
	 In this work, we deal with two achievable schemes: 'XOR coloring' scheme which corresponds to the coloring of the complement of the underlying side information graph $G(V,E)$ or clique cover of $G(V,E)$ and the 'vector XOR coloring' scheme which corresponds to the fractional coloring of $\bar{G}$. We defer the definition and treatment of vector XOR coloring scheme until section \ref{sec:AlgvectorXOR}. First, we focus on the XOR coloring scheme. We define the XOR coloring scheme as follows. 
	\begin{definition}
	       (XOR coloring scheme)  XOR coloring scheme of length $t$ for  $G_d(V,E_d)$ is given by:
	       
	\begin{enumerate}       
	     \item   $t$ linear encoding functions:
	       \begin{equation} 
	                y_i = \sum \limits_{\ell=1}^{n}  G_{i,\ell}  x_\ell, ~ \forall ~1\leq i \leq t 
	        \end{equation}
where $G_{i,\ell} \in \mathbb{F}_2$ and addition is over the binary field $\mathbb{F}_2$. For any $\ell$, $G_{i, \ell}$ is non zero for exactly one $i$.       

             \item  $n$ linear decoding function $\phi_i$ such that: 
\begin{equation}
	                         x_i = \widehat{\phi_i} \left( y_1, y_2 \ldots y_t, N_i \right) = \widehat{\phi_i} \left( y_{\ell}, N_i \right)   
\end{equation}    
for that unique $\ell$, such that $G_{i,\ell} \neq 0$. Here, $N(i)$ is the set of side information packets user $i$ has access to, listed by the directed out-neighborhood of vertex $i$ in the graph $G_d$. In other words, user $i$ needs only that encoded transmission in which its desired packet participates to decode.
\end{enumerate}  \hfill $\lozenge$.
	\end{definition}
	
	For the ICH problem, the side information graph $G_d$ is such that $N(i)= \bigcup \limits_{j: i \in C_j} S_j$. In other words, the side information available to user $i$ is the union of all the caches of the helpers it is connected to.
	
	        It is known that 'XOR coloring' scheme corresponds to a clique cover on the underlying undirected side information graph or the coloring of its complement. Consider one encoded transmission $y_{\ell}$. Every user corresponding to the packet that participates in this should form a clique in $G(V,E)$ or an independent set in the complement of $G(V,E)$ for decoding to be feasible. Each color specifies an encoding function $y_\ell$. Hence, $\chi(\bar{G})$ or $\bar{\chi}(G)$ is equal to the length of the optimal 'XOR coloring' code. Coloring a graph is NP-hard in general.	 
	        
           In the rest of this section, we show that an ICH problem instance can always be reduced to an equivalent ``canonical form'' where each user is connected to a single helper.  Under some conditions, the number of helpers in the equivalent canonical form is linear in the number of helpers of the original network. This will allow us, in the rest of the paper, to restrict our treatment to networks in canonical form, with $k$ helpers, where $k$ is a constant with respect to $n$ .
           	
	We reduce the ICH problem to a problem where every user is connected to a single helper,  by introducing virtual helpers defined 
	as follows: for every possible distinct set $S_{i_1} \bigcup \cdots \bigcup S_{i_m} $, create a virtual helper labeled by 
	$H_{i_1 \ldots i_m}^{\cup}$, with side information $S_{i_1 \ldots i_m }^{\cup} = S_{i_1} \bigcup \cdots \bigcup S_{i_m}$ and neighborhood 
	$\bigcap_{j=1}^m C_{i_j}$.  Clearly, there can be at most $2^k$ such virtual helpers. 
	An important property is that every user is connected to exactly 
	one helper. Namely, user $i$ is connected to the unique virtual helper $H_{i_1 \ldots i_m}^{\cup}$ for which 
	$i \in \bigcap_{j=1}^m C_{i_j}$.
	
	After this operation, two virtual helpers might contain the same side information 
	in which case they can be merged into a single virtual helper. 
	We call this \textit{union expansion}. Note that in the generic case, the expansion requires creation of exponential number of virtual helpers. 
	However, in a practical situation arising in wireless networks, 
	given a planar topology for location of users and helpers, and assuming a communication radius  for each helper, 
	and restricting location of helpers such that any user can be connected to at most $d$ helpers ($d$ could be a small number like $3,4$ due to interference 
	management), the number of virtual helpers that one has to consider in the union expansion 
	is just $O(k)$. We will prove this result in the following lemma.
	
	\begin{lemma}
	Consider a set of $k$ circles on a 2-dimensional plane parametrized 
	by their centers ${\cal P}=\{p_1,p_2 \ldots p_k\}$ and radii $\{r_1,r_2 \ldots r_k\}$. 
	Assume that they form a $d$-ply system, i.e., any point in the plane is covered by at most 
	$d$ circles. Consider the family of subsets ${\cal F} \subseteq 2^{\cal P}$ such that $S \in {\cal F}$ if all the circles in $S$ have a non-trivial intersection. We call such a set $S$ as \textit{intersecting set of circles}. Then $\lvert {\cal F}\rvert = O(k)$.  \hfill $\square$
	\end{lemma}  
	
	\begin{proof}
	Define the intersection graph of the $d$-ply system as $G=(\Pc,E)$ and $(p_i,p_j) \in E$ if circles with centers $p_i$  and $p_j$ intersect. 
	An undirected graph is said to be $\delta$ inductive if there is a numbering on the vertices such that every vertex is connected to at 
	most $\delta$ vertices which are numbered higher than itself.  It is known \cite{miller1997separators} that a $2$-dimensional $d$-ply system is $9d$ inductive. 
	
	Consider a particular circle whose center is  $p_i$. It intersects with at most $9d$ circles which are numbered higher. Let $S$ be a set of circles such that $S \in {\cal F}$ and $p_i \in S$ and if there is a $p_j \in S, ~ p_j \ne p_i$, then $p_j$ is numbered higher than $p_i$.
	We count the number of such sets $S$. An immediate bound would be $\sum_{i=0}^{d-1} {9d \choose i} \leq d  {9d \choose d-1}$. For every circle $p_i$, we count such intersecting sets $S$ that contain $p_i$ but do not contain any circle numbered lower. In this way, one can count at most ${9d \choose d-1}d k$ sets. We have to show that, in this method of counting, every set  $S \in {\cal F}$ is counted at least once. 
	
	In the above lemma, the center of a circle represents a helper and its radius specifies its connectivity range. A $d$-ply system implies that no user can be connected to more than $d$ helpers.
	
	Consider $S \in {\cal F}$. Let circle $p$ have the lowest numbering. This set will be included when counting intersecting sets of circles containing the circle $p$ and no other circle numbered lower.  Hence $\lvert {\cal F} \rvert \leq {9d \choose d-1}d k$. Hence $\lvert {\cal F} \rvert = O(k)$.
	\end{proof}
	
	Note that, although the coefficient $d {9d \choose d-1}$ may be large, the actual number of virtual helpers under union expansion is usually  considerably smaller since some regions of intersection may not contain any users for a given instance of the problem and the side information corresponding to multiple virtual helpers may be identical and hence could be merged further. Due to the above result, it is sufficient (in an order sense with respect to the number of helpers) to treat the problem assuming that a user is connected to exactly $1$ helper. We can re-define our network assuming a constant number (say $k$) of helpers such that their neighborhoods $C_j$ are pairwise disjoint. Since, a user request is not found in the cache of a helper connected to the user, we have $C_j \bigcap S_j = \emptyset$ . There is a slight abuse of notation because $C_j$ is a subset of users and $S_j$ is a subset of packets. From now on, a user and its requested packet is synonymous. Now, we define a variant of Problem \ref{HelperProblem} (ICH) where the network is in canonical form.
		
	\begin{problem}\label{Indexkhelper}
	   (\textit{Index Coding with Disjoint Helpers: ICDH($\{C_i,S_i\}_{i=1}^k$)}) An ICH problem instance $\{C_j,S_j\}$ for $1 \leq j \leq k$ with additional conditions: $\{C_j\}$ pairwise disjoint and $C_j \bigcap S_j = \emptyset$. 
	  \end{problem}
	
	 We are interested in the complexity of computing the optimal 'XOR coloring' scheme for ICDH when $k$ (number of helpers) is a constant. For this we need to analyze some structural properties of $G_d(V,E_d)$ corresponding to the ICDH problem. Lets start with a simple case when $k=2$. 
	\begin{figure}
	 \centering
	  \includegraphics[width=3cm]{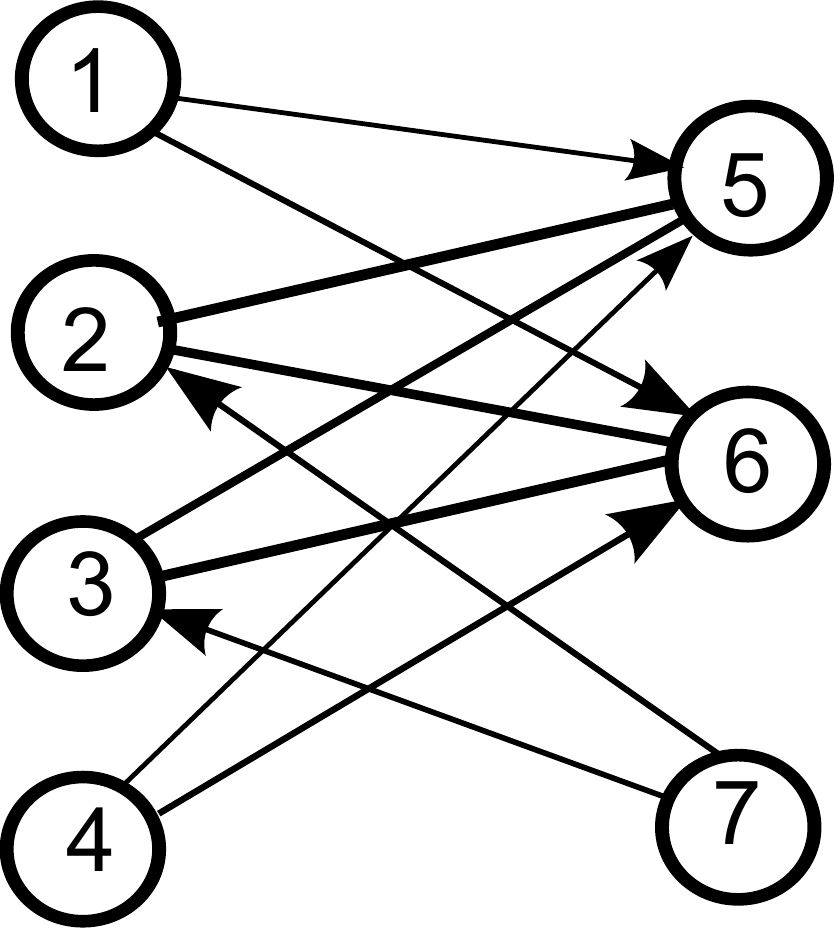}
	  \caption{ $k=2$ helpers. $C_1=\{1,2,3,4 \}, ~ C_2=\{5,6,7\}$. $S_1=\{5,6\}$,$S_2=\{2,3\}$. There is an undirected edge $u(2,5)$ between $2$ and $5$ which represents two directed edge $(2,5)$ and $(5,2)$ in opposite directions. The directed edge $(4,6)$ signifies that user $4$ can see packet $x_6$ due to helper $1$ containing it as side information.}
	  \label{Fig:2helper}
	 \end{figure}    
	 
	\section{$k=2$ helpers case}\label{helper2sec}
	            We consider the ICDH problem for $k=2$. The problem consists of two disjoint sets of users $C_1,C_2$ connected to helpers $1$ and $2$ respectively. Side information contained in helper $1$ is given by $S_1$ and side-information contained in helper $2$ is given $S_2$. To visualize this, let us consider a directed side-information graph, which in this case, is a bipartite graph $G_d(C_1 \bigcup C_2,E_d)$. Since $C_1 \cap S_1 = \emptyset$ and $C_2 \cap S_2 = \emptyset$,  there are directed edges only between the partitions $C_1$ and $C_2$. If $j \in C_2$ and $j \in S_1$, there is a directed edge from every vertex in $C_1$ to the vertex $j$ in $C_2$ signifying that all users in $C_1$ have access to packet $x_j$ desired by user $j$. Two directed edges in opposite directions, i.e. $(i,j)$ and $(j,i)$ are replaced by an undirected edge denoted by $\{i,j\}$. With some abuse of notation, let $(i,j) \in E_d$ represent the directed edge from $i$ to $j$ and by $\{i,j\} \in E_d$ we mean the undirected version. With an example, this is illustrated in  Fig. \ref{Fig:2helper}.
	            
	 We have the following lemma regarding the optimum binary scalar linear index coding solution when $k=2$.
	 \begin{lemma}\label{lem:2helpdecomp}
	                When $k=2$, and each user is connected to exactly one helper, the optimum binary scalar linear code is given by the XOR coloring scheme and its length is $| C_1|+| C_2 | -  \min \{ | S_1 \cap C_ 2 |, | S_2 \cap C_1 | \}$.  \hfill $\square$
	\end{lemma}
	 
	 \begin{proof}
	    Let $C_1^{\rm{out}} \subseteq C_1$ be the set of users which are not in $S_2$. Similarly, let $C_2^{\rm{out}} \subseteq C_2$ be the users whose packets are not present as side information in helper $1$. Observe that component induced by $C_1^{\rm{out}} \cup C_2^{\rm{out}}$ has no directed/undirected edges and edges are only directed out of this component in the original graph.  Let $E\rq{}$ be the set of edges that are directed out from the component induced by $C_1^{\rm{out}} \cup C_2^{\rm{out}}$. Hence one can decompose the original bi-partite graph into two parts as follows: $G_1(C_1^{\rm{out}} \cup C_2^{\rm{out}}, \emptyset)  \rightarrow G_2(C_1\backslash C_1^{\rm{out}} \cup C_2\backslash C_2^{\rm{out}}, E_d \backslash E\rq{})$ where the \lq{}$\rightarrow$\rq{} indicates that there are no edges coming into $G_1$ from $G_2$ but possibly there are edges in the other direction. An example is given in Fig. \ref{Fig:2helper} for illustration.  The decomposition for the example given in Fig. \ref{Fig:2helper} is provided in Fig. \ref{Fig:2helperdecomp}. In the example, $C_1^{\rm{out}} = \{1,4\}$ and $C_2^{\rm{out}}=\{7\}$.
	    
	     Consider the graph decomposed as $G_1 \rightarrow G_2$.  The second component $G_2$ is a complete bipartite graph $K_{ C_1 \cap S_2 ,  C_2 \cap S_1  }$ with only undirected edges where $K_{M,N}$ denotes a complete bi-partite graph between disjoint vertices of the set $M$  and disjoint vertices of the set $N$.  In other words, in the remaining graph $G_2$, given any two vertices $i,j$ from different partitions, there is an undirected edge $\{i,j\} \in E_d \backslash E\rq{}$. This is because, every vertex in one partition is contained in the side information set of the other partition. Otherwise that vertex will belong to $C_1^{\rm{out}} \cup C_2^{\rm{out}}$ since it would be a vertex not \lq{}seen\rq{} by the other partition.
	     
	     For a complete bi-partite graph, the independent number $\alpha(G_2)$ and the clique cover number $\bar{\chi}(G_2)$ are identical because a bi-partite graph is a perfect graph \cite{bollobas}. In this case, the independence number equals the size of the maximum of the two partitions , i.e. $\max \{  | S_1 \cap C_ 2 |, | S_2 \cap C_1 |  \}$.  It is known that for an undirected side information graph $G$, $\alpha(G) \leq \ell(G) \leq \bar{\chi}(G)$ where $\ell(.)$ is the length of the optimum binary scalar linear index code of the undirected graph $G$. Hence, $\ell(G_2) = \alpha(G_2)$. From the decomposition of $G_1 \rightarrow G_2$, since no vertex in $G_1$ is contained as side information in $G_2$, it is known that \cite{bar2011index} $\ell(G_d)=\ell(G_1)+\ell(G_2)$ where $\ell (G_d)$ is the optimal scalar linear index code for $G_d$. In this case, $\ell(G_1)= | G_1 | = | C_1^{\rm{out}} \cup C_2^{\rm{out}} |$ since $G_1$ has no edges.  Hence, the result follows. In the example given in Fig. \ref{Fig:2helperdecomp}, the transmissions that can be saved corresponds to the complete bipartite graph between $\{2,3\}$ on one side and $\{5,6\}$ on the other. One could transmit $x_2 + x_5$ and $x_3+x_6$ or $x_2+x_6$ and $x_3+x_5$. This corresponds to the XOR coloring solution of the bipartite graph	 
	 \end{proof}
	 \begin{figure}
	 \centering
	  \includegraphics[width=6cm]{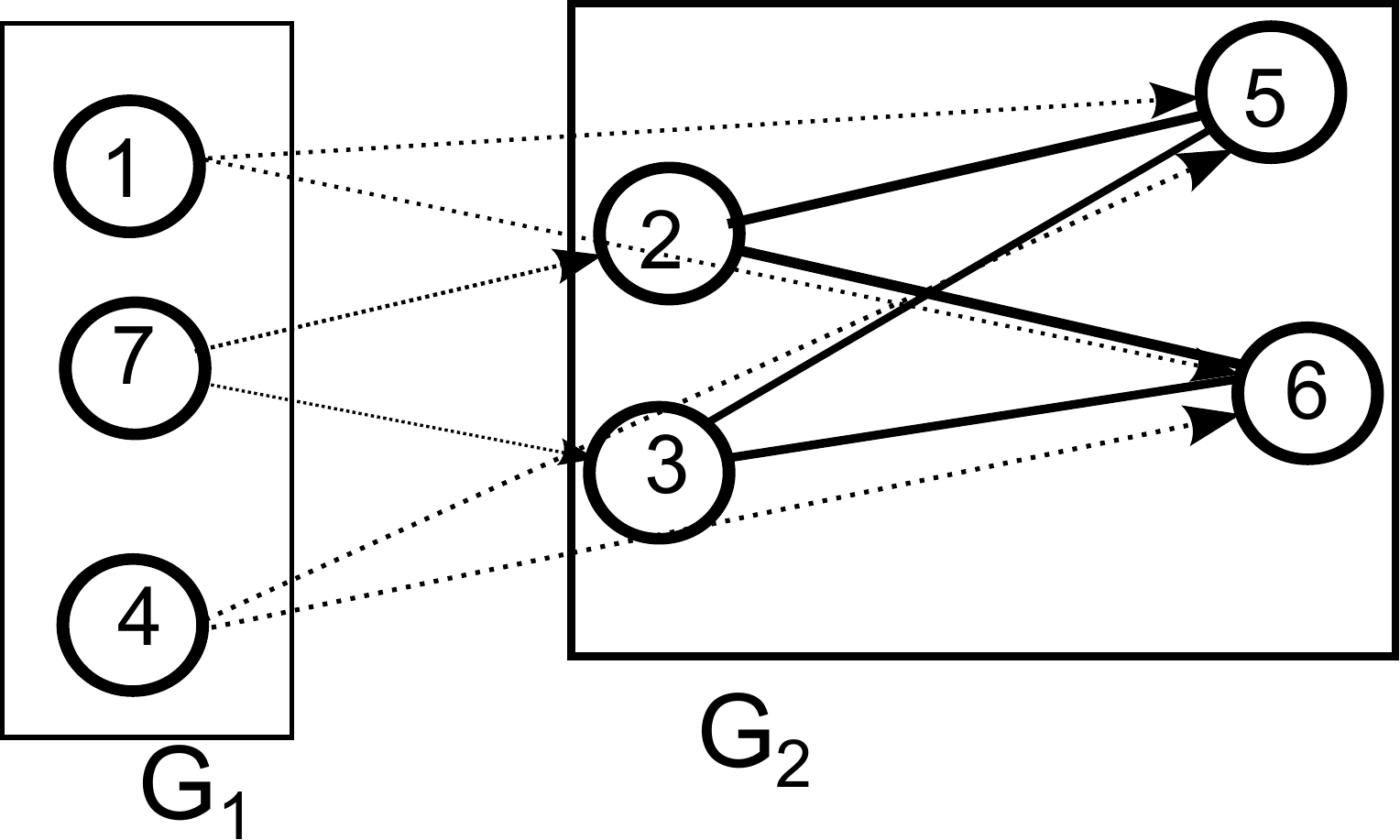}
	  \caption{The original graph from Fig. \ref{Fig:2helper} decomposed as in the proof of Lemma \ref{lem:2helpdecomp}}.
	  \label{Fig:2helperdecomp}
	 \end{figure}  
 In a nutshell, for $k=2$, the bi-partite graph exhibits a special bi-partite structure that makes the XOR coloring solution the optimal binary scalar linear index coding solution. Using the ideas from this section, we now investigate the structure of the side-information graph for the case $k \geq 3$.	 
	  
	 \section{$k \geq 3$ Helpers case}\label{helper3sec}
	  For a given $k$, the ICDH problem input consists of neighborhoods $\{C_i\}_{i=1}^k$ and the side information sets $\{S_i\}_{i=1}^k$ such that $C_i \bigcap C_j =\emptyset$. In this case, the
	 side information graph is given by a directed $k$-partite graph $G_d(\bigcup \limits_{i} C_i, E_d)$. As before , an edge is undirected when edges in both directions are present between two vertices. We isolate vertices which are not contained in any helper into $\bigcup \limits_{i} C_i^{\rm{out}}$. This allows us to define a very similar decomposition: $G_1\left(\bigcup \limits_{i} C_i^{\rm{out}}, \emptyset\right) \rightarrow G_2\left(\bigcup \limits_{i} C_i \backslash \bigcup \limits_{i} C_i^{\rm{out}}, E_d \backslash E\rq{}\right) $ as before where $E\rq{}$ consists of edges directed from $G_1$ into $G_2$ as before. The decomposition for $k=3$ case is illustrated in Fig. \ref{Fig:helper3}.
	 
	 A new feature is that, in the second graph $G_2$, there are directed edges too. In the example given in Fig. \ref{Fig:helper3}, directed edges appear in $G_2$. $(4,5)$ exists but there is no directed edge from partition $2$ to node $4$. In particular, edge $(5,4)$ is not present. We are interested in $\chi(\bar{G})$ where $G$ is the underlying undirected side information graph corresponding to the problem instance $G_d$. $G$ is the disjoint union of $G_1^u$ and $G_2^u$ where $G_i^u$ is the undirected version of $G_i$ obtained by retaining only undirected edges. Observe, that there are no undirected edges between $G_1$ and $G_2$. Hence, $\chi \left(\bar{G} \right)= \chi \left( \bar{G}_1^u \right) +\chi \left(\bar{G}_2^u \right)$. Since $G_1^u$ has no edges it is sufficient to color $\bar{G}_2^u$. In fact, let us assume $C_i^{\rm{out}}=\emptyset,~\forall i$ because even if they are not empty they are collected in the edgeless graph $G_1^u$. Then, $G_2^u$ has the following structure:
	 \begin{lemma}\label{lem:Completebip}
	 The vertex induced subgraph of $G^u_2$, induced by $C_i$ and $C_j $ (partition $i$ and $j$ in $G^u_2$), has a complete bipartite graph $K_{C_i \cap S_j,  C_j \cap S_i}$ and no other edges. Some vertices which are not part of the bi-partition, but present in the induced sub-graph, have degree $0$ (no edges incident on them). 
	\end{lemma}  
	\begin{proof}
	          Since the induced sub-graph is between user indices belonging to the neighborhood of two helpers and only undirected edges are considered, the result follows from the structure of the undirected side information graph for $k=2$ helpers. The argument is found in the proof for Lemma \ref{lem:2helpdecomp}. Note that for the case of $k=2$ helpers, $G^u_2$ was just a complete bipartite graph between partitions $1$ and $2$. The same argument extends to partition pairs $i$ and $j$ when $k>2$. When $k>2$, for the subgraph of $G^u_2$ induced by the vertices in the partition $i$ and $j$, there may be additional nodes in either partition that have degree $0$. The extra nodes may arise because nodes connect to other helper partitions but not $i$ and $j$.  For illustration, compare and contrast $G_2$ in Fig. \ref{Fig:2helperdecomp} and Fig. \ref{Fig:helper3}. In Fig. \ref{Fig:helper3}, node $4$ participates in the second $G_2^u$. However, considering the subgraph induced by the first two partitions, it has degree $0$.  
	\end{proof}
	
        Henceforth, we use the phrase \textit{complete bipartite structure} to denote this property of $k$-partite graphs. For the remainder of this work, we focus on coloring a $k$-partite undirected graph $G^u$ in which every vertex has degree at least $1$ and has the complete bipartite structure between any two partitions. 

	  \begin{figure}
	 \centering
	  \includegraphics[width=11cm]{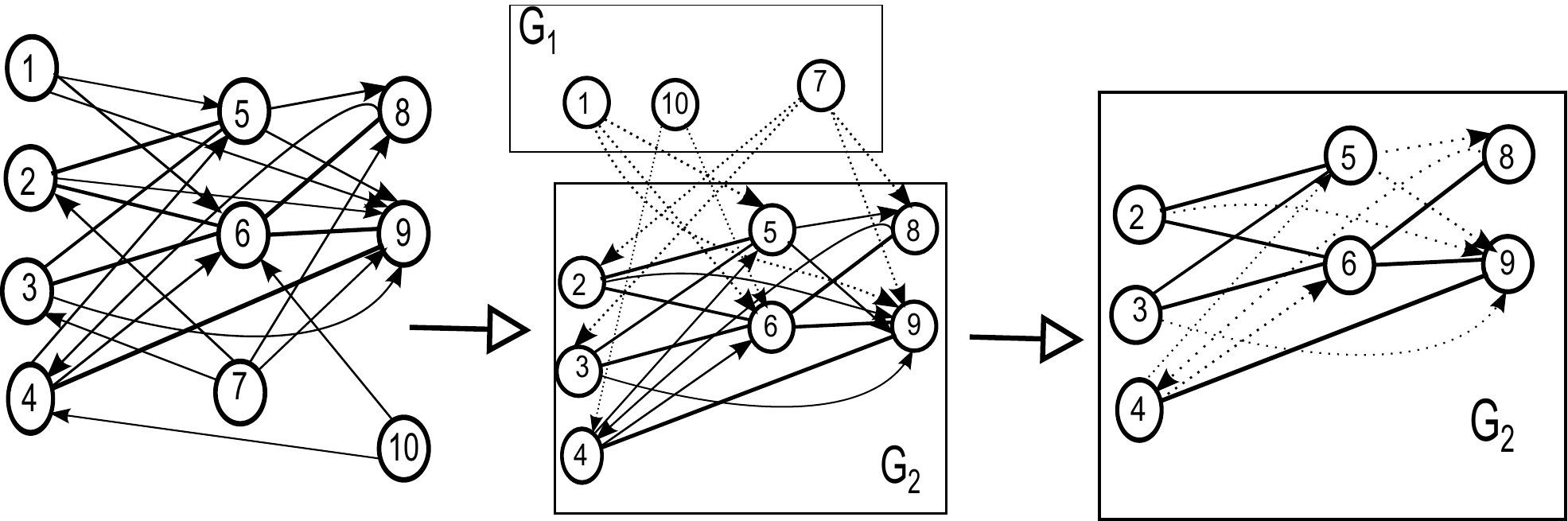}
	  \caption{Decomposition for $k=3$ helpers case. Notice that in $G_2$ there are directed edges. In $G_2^u$, any two partitions induce a complete bipartite graph.}
	  \label{Fig:helper3}
	 \end{figure}     
	 
	 \textbf{Remark:} Complement of a tripartite undirected graph is a triclique. It is known that coloring a triclique is  NP-complete \cite{venkatesan1996approximate}.  Hence, for a general $3$-partite graph coloring the complement is hard. However, we will show that the complete bipartite structure helps us find an algorithm that runs in time $\rm{poly}(n)$ for constant $k$.             
	 
               \section{Category Graph and Multicoloring}\label{categorysec}
	        	  
	           The complete bipartite structure motivates grouping vertices into bins each of which has a label. Consider a set of partitions $\{i_1,i_2 \ldots i_p\} \subseteq \{2,3,4 \ldots k\}$ ($p \leq k-1$). Let $\{j_1,j_2 \ldots j_{k-1-p}\} =\{ i_1,i_2 \ldots i_p\}^c $, where the complement is with respect to the set $[k]-\{1\}$ and $[k]$ denotes the set of numbers from $1$ to $k$. A \textit{category}, denoted by the label $V_{1 \rightarrow i_1i_2 \ldots i_p j_1^cj_2^c \ldots j^c_{k-1-p} }$,  contains all vertices in partition $1$ which are connected to at least one node in each partition belonging to $\{i_1 \ldots i_p\}$ but not connected to any node in any partition belonging to $\{j_1 \ldots j_{k-p-1} \}$. 
	
	There are $ k 2^{k-1}$ such categories with labels $\{V_{i \rightarrow \mathbf{p}\mathbf{s}^c } \}, \forall i \in [k], \forall  \mathbf{p},\mathbf{s}: \mathbf{p}+\mathbf{s} = [k]-\{i\} , \mathbf{p} \bigcap \mathbf{s} = \emptyset $. Here, with some abuse of notation, $\mathbf{p}$ is a string of indices (order within the string does not matter) which corresponds to the subset of indices indicated by $\mathbf{p}$ and string $\mathbf{s}$ represents the subset of indices which occur with a complement in the label for the category. $'+'$ means string concatenation (or equivalently disjoint union).  Of these, categories of the form $V_{i \rightarrow \mathbf{s}^c}$ where $\mathbf{s} = [k]-\{i\}$ have zero vertices because these are vertices with degree $0$ and hence are removed prior to coloring by assumption. Hence, there can be potentially $k(2^{k-1}-1)$ categories.      
	    
	         Because of the complete bipartite structure between the partitions, the categories have the following properties:
        \begin{enumerate} 	         
	        \item  All vertices in a given category are equivalent (they are connected to the same set of vertices in the rest of the graph) with respect to connections in $G^u$.
	        \item There is a complete bipartition between vertices belonging to a category with label $V_{i\rightarrow j \ldots}$ and vertices belonging to a category with label $V_{j \rightarrow i \ldots}$. Here, $\ldots$ represents any valid string of indices (both substitutions can be different) that could be substituted in both category labels. There are no edges between vertices belonging to any other pair of categories. 
 	\end{enumerate}
	
	The properties above motivate the following definition:        
	        \begin{definition}
	            A \textit{category graph} ${\cal G}\left( {\cal V}, {\cal E} \right)$ is an undirected graph with ${\cal V}$ containing all category labels and an edge between labels if and only if all vertices from both categories participate in a complete bipartition in $G^u$.  For each category label vertex $v$, let the weight $w(v)$ be the number of vertices which belong to that category. 
	       \end{definition} 
	        
	        The above definition is specific to $k$-partite graphs with complete bipartite structure between partitions. Note that ${\cal G}$, except for the weights $w(.)$, is fully determined by $k$ and therefore independent of the ICDH problem input.  We have $\sum_{v \in {\cal V}} w(v) =n$. For illustration, category graph for $k=3$ is given in Fig. \ref{Fig:CategoryGraph}.
	 
	        One can expand the category graph by replacing every category node $v$ by a graph of $w(v)$ disconnected vertices and connect every vertex in category $v$ to every vertex in category $u$ if and only if $(u,v) \in {\cal E}$. This expansion results in the original undirected $k$-partite graph $G^u$. Let us consider the problem of coloring $\bar{G}^u$. We have the following definition and lemma that characterizes coloring $\bar{G^u}$ in terms of multi coloring the complement of the category graph. 
	
	\begin{definition}
	    An $r$-multi coloring of an undirected $G({\cal V},{\cal E})$ with a positive integral weight function $w:V\rightarrow {\mathbb Z}^{+}$ is assigning $r$ colors in total such that $w(v)$ different colors are assigned to vertex $v$ and if $(u,v) \in {\cal E}$ then colors assigned to $u$ and $v$ must all be different. \hfill $\lozenge$
	\end{definition}          
	\begin{theorem}\label{Thm:Multicoloring}
	   Optimal coloring of the complement of a $k$ partite undirected graph $G^u$, with a complete bi-partite structure between any two partitions, is equivalent to optimal multi coloring of the complement of its category graph, i.e. $\bar{G}({\cal V},{\cal E})$  
	\end{theorem}   
	\begin{proof}
	        In the complement of $G^u$, two vertices belonging to the same category will be connected to each other. This implies that the vertices belonging to a single category will form a clique in the complement of $G^u$. This introduces the weight function and multi coloring constraint on every category node on the complement graph $\bar{G}({\cal V},{\cal E})$. Two vertices in different categories will be connected in $\bar{G}$ if and only if the corresponding category labels are connected in $\bar{{\cal G}}$. This introduces multi coloring constraints across different categories in the complement of the category graph. 
	 \end{proof}
	
	  We now exhibit another special feature of $G^u$ with complete bipartite structure property when $k=3$. We show that the graph is perfect when $k=3$. We begin with some results about perfect graphs.
	
	\begin{figure}
	 \centering
	  \includegraphics[width=9cm]{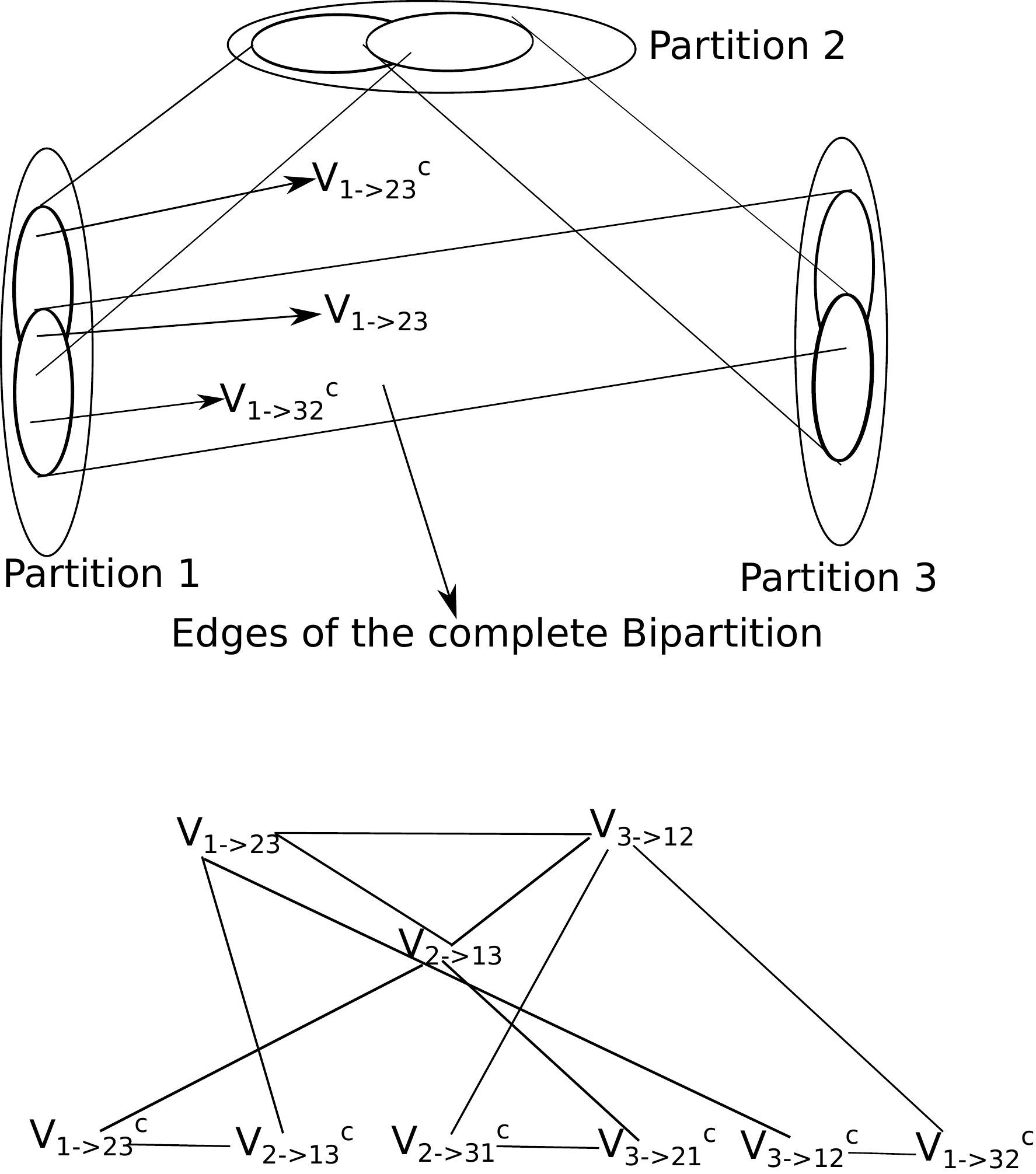}
	  \caption{1. An undirected $3-$ partite graph $G^u$ considered in this work between partitions characterized by different helpers and 2. The category graph of the original graph}
	  \label{Fig:CategoryGraph}
	 \end{figure}
	
	\begin{figure}
	 \centering
	  \includegraphics[width=9cm]{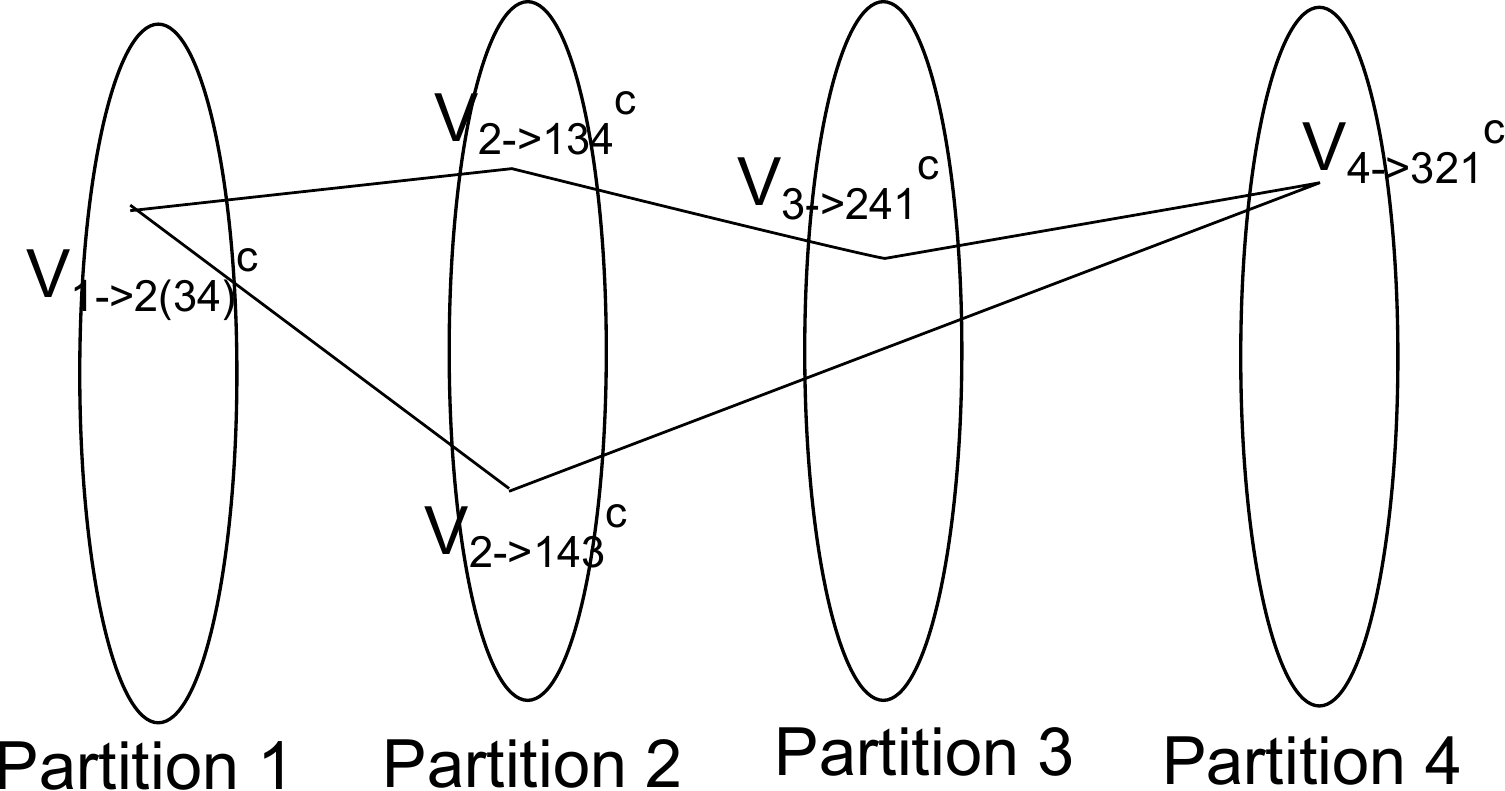}
	  \caption{A cycle of length $5$ which is an induced subgraph when $k \geq 4$}
	  \label{Fig:Counterexample}
	 \end{figure}
	
	\begin{definition}
	        An undirected graph $G^u$ is perfect if for every induced subgraph $H^u$, the coloring number $\chi(H^u)$ is equal to the maximum clique number, denoted $\omega(H^u)$. \hfill $\lozenge$
	\end{definition}
	
	\begin{lemma}\label{lem:ComplementPerfect}
	  \cite{bollobas} An undirected graph is perfect if and only if its complement is perfect. \hfill $\square$
	\end{lemma}
	
	\begin{theorem}\label{thm:StrongPerfectGraph}
	   \cite{chudnovsky2006strong} An undirected graph $G^u$ is perfect if and only if no induced subgraph of $G^u$ is an odd cycle of length at least five (called an odd hole) or the complement of an odd hole (anti hole). \hfill $\square$
	\end{theorem}
	
	\begin{lemma}\label{lem:replacementlemma}
	       \cite{bollobas} A graph obtained from a perfect graph by \textit{replacing} its vertices by perfect graphs is perfect. In a replacement, if $(u,v)$ is an edge in the original graph, each vertex of a graph which replaces $u$ will be connected to every vertex of the graph that replaces vertex $v$. \hfill $\square$
	\end{lemma}
	
	\begin{theorem}\label{Thm:perfect3}
	              Any $3$-partite undirected graph $G^u$ possessing a complete bipartite structure between two partitions is perfect. For $k$-partite graphs, possessing the same structure with $k \geq 4$, the graph in general is not perfect.  \hfill $\square$
	\end{theorem}
	\begin{proof}
	              Consider the category graph of the $3$-partite graph given in Fig. \ref{Fig:CategoryGraph}.  It is easy to observe that it does not contain odd cycle of length $5$ or more. The obstruction to any such attempt comes from the triangle formed by $V_{1\rightarrow 23},V_{2 \rightarrow 13},V_{3 \rightarrow 12}$.  Now, consider odd anti holes. Complement of a cycle of length $5$ happens to be a cycle of length $5$ which is ruled out. Complement of a cycle of length $7$ or more has at least ${7 \choose 2}-7 = 14$ edges.  
	              Nevertheless, the total number of edges in the category graph is only $12$ . Hence, by Theorem \ref{thm:StrongPerfectGraph}, the category graph for $k=3$ case is perfect. Now, to get the original $3$-partite graph from the category graph, one can expand the category graph by replacing every category vertex $v$ with a completely disconnected graph of size $w(v)$. A completely disconnected graph is perfect. Therefore, by Lemma \ref{lem:replacementlemma}, the conclusion of the theorem follows.	
	              
	            Next, we exhibit a $4$-partite graph with the complete bipartite structure property between every two partitions, which is a cycle of length $5$. This has been shown in Fig. \ref{Fig:Counterexample}. It consists of one vertex each from the categories $V_{1\rightarrow 23^c4^c},V_{2\rightarrow 134^c},V_{3\rightarrow 24^1c},V_{1\rightarrow 4321^c}$ and $V_{2\rightarrow 143^c}$. Hence, not all $k$-partite graphs with the complete bipartite structure are perfect when $k \geq 4$  in general.    
	\end{proof}
	\begin{corollary}
	      For, $k=3$, considering the underlying undirected side information graph of the ICDH problem, XOR coloring scheme gives the optimal binary scalar linear code. \hfill $\square$
	\end{corollary} 
	\begin{proof}
	      Note that, from the original graph $G_2$, directed edges were deleted to create a $k$-partite undirected graph $G_2^u$ (for notation refer Section \ref{helper3sec}). From Theorem \ref{Thm:perfect3}, the graph $G_2^u$ and its complement are perfect (Lemma \ref{lem:ComplementPerfect}) for $k=3$. As noted before, for perfect undirected graphs, coloring is the optimal index code \cite{blasiak2010index}. With $\ell (G^u)=\ell(G_1^u)+\ell(G_2^u)$ as $G_1^u$ and $G_2^u$ are disconnected. This proves the claim. 
       \end{proof}	       
	      There is a polynomial time algorithm to color perfect graphs \cite{grotschelgeometric}. As a low complexity alternative to that, we provide a greedy coloring (multi-coloring) algorithm which can optimally color, in polynomial time, the complement of the $3$-partite graph with the complete bi-partite structure. 
	
	   For an undirected graph, greedy coloring comprises of an ordering of the vertices and assigning colors (or equivalently numbers starting from $1$) sequentially. For every vertex in the order, the smallest color (number), except the ones already assigned to its neighbors processed previously in the ordering, is assigned. If no color previously can be used, then a new color is assigned (the smallest unused number). The number of colors is basically the largest number assigned to any vertex at the end of coloring. Finding the optimal order, given a graph, is intractable in general. In the case of $k=3$, we give an optimal ordering for multi coloring the complement of the category graph.  
	
	\begin{theorem}
	           An optimal ordering for coloring when $k=3$ consists of processing all vertices of one category at the same time in any order. The order for processing categories is as follows: $V_{1\rightarrow23^c} \rightarrow V_{2\rightarrow 1 3^c} \rightarrow V_{2\rightarrow31^c} \rightarrow V_{3\rightarrow 21^c} \rightarrow V_{3\rightarrow12^c} \rightarrow V_{1\rightarrow32^c} \rightarrow V_{1\rightarrow23} \rightarrow V_{2\rightarrow 1 3} \rightarrow V_{3\rightarrow12} $.  \hfill $\square$
	\end{theorem} 
	\begin{proof}
	          Let us consider the category graph, denoted $G^u$, for $k=3$ given in Fig. \ref{Fig:CategoryGraph}. The proof exploits the structure of this graph. One needs to multicolor the complement of this graph. Instead of considering the edges of the complement, we will consider the constraints from coloring on the category graph itself. For coloring any category, the colors from the categories not connected to this category are forbidden and any remaining colors can be used. Let us denote the vertices (or multicolors needed) in each category by the weight function $w(v), v \in {\cal V}$. Let us consider the subgraph induced by $V_{i \rightarrow j k^c}, \forall ~i \neq j \neq k$. Denote it by $G_1^u$. In $G_1^u$, any category $V_{i \rightarrow jk^c}$ can be colored only by colors from $V_{j \rightarrow ik^c}$ and vice versa. Hence, the forbidden category sets for both are identical in $G_1^u$ and comprise the remaining category sets in $G_1^u$. Therefore, distinct colors used for coloring both the categories in any optimal coloring of $G_1^u$ has to be at least $\max (w(V_{i \rightarrow jk^c}),w(V_{j \rightarrow ik^c}))$. It is easy to see that following the greedy order given in the theorem, with respect to $G_1^u$, achieves this lower bound and hence optimal. According to the order given in the theorem, it is enough to show that after greedy optimal coloring of $G_1^u$, coloring category sets $V_{i \rightarrow jk}$ (in any order) will be optimal for $G^u$. This is because , without loss of generality, the partitions $1,2,3$ can be interchanged in labels as we do not assume anything about the  weights of different categories. 
	
	        Let $p$ be the optimal number of colors used by the optimal greedy algorithm on $G_1^u$. Now assuming greedy ordering in the theorem, the total colors required for the entire graph $G^u$ is given by:
	\begin{align*}
	\rm{Colors ~ used ~ in ~ greedy} = p+\max \limits_{i, j \neq k \neq i} \left(V_{i \rightarrow jk} - \left(p - w\left(V_{i \rightarrow jk^c}\right) \right.\right. \nonumber\\
	\left.\left. - w\left(V_{k \rightarrow ij^c}\right) - \max \left(w\left(V_{j \rightarrow ki^c}\right), w\left(V_{k \rightarrow ji^c}\right)\right) \right)\right)^{+} \nonumber
	\end{align*}
	 where $(x)^{+}=\max (x,0)$. Categories $V_{i \rightarrow jk},~\forall i \neq j \neq k$ form a triangle (clique). From $p$ colors used for $G_1^u$, forbidden colors for the category $V_{i \rightarrow jk}$ (i.e. colors of categories $V_{i \rightarrow jk^c}$, $V_{k \rightarrow ij^c}$ distinct colors used for the connected pair $V_{j \rightarrow ki^c},V_{k \rightarrow ji^c}$- all three that do not share any mutual colors) are removed. The remaining colors from $p$ colors are used to color $V_{i \rightarrow jk}$. If still new colors are required for any of $V_{i \rightarrow jk}$, then they can be shared among the three categories that form the triangle. This freedom of sharing is the reason for the outer $\max()$ function. $(.)^{+}$ appears because, after using up non-forbidden colors from $G_1^u$, we may not need any new colors for a particular category $V_{i \rightarrow jk}$. The $\max$ function inside follows from the lower bound and its achievability due to greedy coloring for the graph $G_1^u$, as explained previously for the pair $V_{j \rightarrow ki^c},V_{k \rightarrow j i^c}$. 
	
	      Now assume any other arbitrary coloring order (possibly in which order among vertices in same category may also matter) which is optimal. After coloring the entire graph, consider the subgraph $G_1^u$. Since greedy is optimal for this graph, the new optimal coloring must have required $p+\delta$ colors where $\delta \geq 0$ for $G^u_1$. Now, using the same forbidden colors argument as in the preceding paragraph, the total number of colors used in this coloring would be at least: 
	 \begin{align*}
	(p+\delta)+\max \limits_{i, j \neq k \neq i} \left(V_{i \rightarrow jk} - \left(p+\delta - w\left(V_{i \rightarrow jk^c}\right)- \right.\right. \nonumber\\ 
	\left.\left. w\left(V_{k \rightarrow ij^c}\right) - c\left(V_{j \rightarrow ki^c}, V_{k \rightarrow ji^c}\right) \right)\right)^{+} 
	\end{align*}
	 where $(x)^{+}=\max (x,0)$ where $c (v,u)$ represents the number of distinct colors used for multicoloring vertices $u$ and $v$ in the new coloring. In the graph $G_1^u$, $V_{j \rightarrow ki^c}, V_{k \rightarrow ji^c}$ have identical forbidden sets namely, the rest of the categories in $G_1^u$. Therefore,
	\begin{equation*}
	 c\left(V_{j \rightarrow ki^c}, V_{k \rightarrow ji^c}\right) = \max \left(w\left(V_{j \rightarrow ki^c}\right), w\left(V_{k \rightarrow ji^c}\right)\right) + \delta_{i^c}
	\end{equation*}
	where $ \delta_{i^c} \geq 0$. Hence, $ -\delta \leq \delta_{i^c} - \delta $.
	   Therefore, we have the following chain of inequalities:
	           \begin{align} 
	                    \rm{Colors ~ used} & \geq (p+\delta)+\max \limits_{i, j \neq k \neq i} \left(V_{i \rightarrow jk} - p + w\left(V_{i \rightarrow jk^c}\right)+ \right. \nonumber\\
	   &\left.  w\left(V_{k \rightarrow ij^c}\right) + \max \left(w\left(V_{j \rightarrow ki^c}\right), w\left(V_{k \rightarrow ji^c}\right)\right) + \delta_{i^c} - \delta \right)^{+} \nonumber \\
	                     & \geq (p+\delta)+\max \limits_{i, j \neq k \neq i} \left(V_{i \rightarrow jk} - p+ w\left(V_{i \rightarrow jk^c}\right)+ \right. \nonumber\\
	 & \left. w\left(V_{k \rightarrow ij^c}\right) + \max \left(w\left(V_{j \rightarrow ki^c}\right), w\left(V_{k \rightarrow ji^c}\right)\right) -\delta \right)^{+}  \nonumber \\ 
	                     & \geq (p+\delta)+\max \limits_{i, j \neq k \neq i} \left(V_{i \rightarrow jk} - p+ w\left(V_{i \rightarrow jk^c}\right)+ \right. \nonumber\\
	  &\left. w\left(V_{k \rightarrow ij^c}\right) + \max \left(w\left(V_{j \rightarrow ki^c}\right), w\left(V_{k \rightarrow ji^c}\right)\right) \right)^{+} -\delta \label{Eqn:colorsgreedy} \\ 
	                     &  = \rm{Colors ~ used ~ in~ greedy} \nonumber
	           \end{align}
	  Eqn. (\ref{Eqn:colorsgreedy}) is true because $\max (x-\delta,0) \geq \max(x,0) - \delta$.
	    This means that one uses more colors than in the greedy algorithm by coloring it any other way. Hence the proposed order is optimal.          
	\end{proof}
	   
	\textit{Remark:} It is easy to check that, given the ICDH problem, labeling every vertex with the category labels takes at most $O(n^2)$ time. And ordering them would also take $O(n^2)$ time. Any vertex possessing category label with one complement, i.e. $V_{i \rightarrow jk^c}$, has to be placed before all the vertices with no complement, i.e. $V_{i \rightarrow jk}$, and the rest can be in any order without loss of generality for the greedy algorithm to work. Therefore, it takes $O(1)$ time to decide the relative order between two vertices. Greedy coloring would take at most $O(n^2)$ time given the ordering. The greedy coloring takes time polynomial in $n$ which is in sharp contrast with the intractability of coloring for general tri-cliques.
	
	 \section{Integer Programming Formulation for $k \geq 4$ }\label{intprogsec}
	              It is  unclear if a similar greedy order of processing categories can be found out for category graphs for $k \geq 4$. At least, if there is such a provable greedy ordering for any $k$, then deciding the relative order between two vertices of different labels must take at least time exponential in $k$. If it does not, when $k=n$ , it would lead to polynomial time coloring for the complement of a general undirected graph. We provide an integer programming formulation for the multi-coloring problem on the catgeory graph to illustrate that, for constant $k$, one can optimally color in time polynomial in $n$. We provide a specific algorithm based on Graver Basis which can utilize some properties in this setting.
	
	             For $k \geq 4$, the problem reduces to multicoloring the complement of the category graph ${\cal G}\left({\cal V},{\cal E}\right)$ with the positive integral valued weight function $w:{\cal V} \rightarrow \mathbb{Z}^{+}$ (Theorem \ref{Thm:Multicoloring}). The structure of the category graph and its size is fixed for given $k$. Only the weights assigned to each categories, representing the number of vertices in each category, varies depending on the problem instance. Let ${\cal C}$ be the set of all cliques in the category graph. Note that, a clique in the category graph is an independent set in the complement. Hence, we can impose coloring constraints on the category graph itself. Let us denote a clique in ${\cal G} $ by $S=\{i_,1i_2,\ldots i_m\}$ where $i_j \in {\cal V}$ comprise the category vertices in the clique. Let $c(S), ~\forall S \in {\cal C}$, denote the number of common colors assigned to all category vertices comprising the clique. These colors are not used elsewhere in the coloring.  Now, we write multi coloring on $\bar{{\cal G}}$ as the following integer linear program. This formulation is similar to that in \cite{mehrotra2007branch}. 
	
	       \begin{align} \label{Eqn:integcol} 
	                      & \min \sum \limits_{S \in {\cal C}} c(S) \nonumber \\
	                       & \rm{subject~to~}   \sum \limits_{S: v \in S } c(S) = w(v), \forall v \in {\cal V} \nonumber \\
	                          &              c(S) \geq 0, ~ c(S) \in \mathbb{Z} ~ \forall S \in {\cal C}   
	        \end{align}
	
	The linear integer program is of the following form:
	  \begin{align}\label{IntegerProgCol}
	     & \min \mathbf{1}^T\mathbf{c} \nonumber \\ 
	      & \rm {subject~ to~} \mathbf{A}\mathbf{c} =\mathbf{w} \nonumber \\
	      &  \mathbf{c} \geq 0,~ \mathbf{c} \in \mathbb{Z}^{f_2(k)}      
	  \end{align}
	 where $\mathbf{c}$ is the vector comprising of entries $\{c(S)\}_{\forall S \in {\cal C}}$ and $\mathbf{A}$ is matrix which in turn depends only on the structure of the category graph which depends only on $k$. Let the number of rows of $\mathbf{A}$ be $f_1(k)$. Let the number of columns of $\mathbf{A}$ be $f_2(k)$. The number of rows of $\mathbf{A}$ is the total number of categories which is $k\left(2^{k-1}-1 \right)$. The length of vector $\mathbf{c}$ is the total number of cliques in the category graph ${\cal G}$. A very loose bound of $k^k 2^{k^2}$ can be established using that fact that the largest clique size is $k$ for the category graph. $\mathbf{w}$ represents vector of weights $w(v),~\forall v \in {\cal V}$. The magnitude of entries in $\mathbf{w}$ is bounded by $n$. Hence, it needs $\log n$ bits to represent each entry. When the number of constraints and variables is fixed, it is possible to solve the integer program in time polynomial in $n$ by methods in \cite{lenstra1983integer}. 
	
	 In this case, the matrix $\mathbf{A}$ is fixed and does not change for a given $k$. Exploiting this through an approach of computing Graver basis of $\mathbf{A}$ to solve the integer program, we have the following result: 
	  \begin{theorem} \label{Thm:XORcol}
	   The XOR coloring scheme for the ICDH problem with constant number of helpers ($k$) and $n$ users can be computed in time $O(n \log n)$ by solving Problem (\ref{IntegerProgCol}). 
	\end{theorem}
	 \begin{proof}
	     The proof is given in the appendix.
	 \end{proof}
	 	
	\section{Extension to Vector XOR coloring}\label{sec:AlgvectorXOR}
	  In this section, we extend the above tractability results to another achievable scheme that we call \textit{vector XOR coloring}. This scheme is based on fractional coloring of the complement of the undirected side information graph $G$ first introduced in \cite{blasiak2010index}. This scheme is a vector linear scheme. Therefore, let us assume that the message desired by the $i$-th user is a packet consisting of $p$ bits $\left[x_{i1}~x_{i2}~x_{i3} \ldots x_{ip} \right] \in \mathbb{F}_2^{1 \times p}$. This assumption models the following scenario: The message desired by user $i$ is split into $p$ sub-packets of equal size and every bit of sub-packet $m$ undergoes the same encoding (XORed similarly) as the corresponding bit $x_{im}$ in the encoding scheme.
	  
	   \begin{definition}
	       A Vector XOR coloring scheme of broadcast rate $\frac{t}{p}$ for  $G_d(V,E_d)$ is given by:
	       
	\begin{enumerate}       
	     \item   $t$ linear encoding functions:
	       \begin{equation} \label{Eqn:Colorenc}
	                y_m =\sum \limits_{i=1}^{n} \sum \limits_{j=1}^{p}  G_{m,ij}  x_{ij}, ~ \forall ~1 \leq m \leq t 
	        \end{equation}
where $G_{m,ij} \in \mathbb{F}_2$ and addition is over the binary field $\mathbb{F}_2$. For any $i,j: 1 \leq i \leq n,~ 1 \leq j \leq p$, $G_{m,ij}$ is non zero for exactly one $m$. Every sub-packet participates in only one of the transmissions.       

             \item  $np$ linear decoding functions $\widehat{\phi_{ij}}, ~\forall 1 \leq i \leq n,~ 1 \leq j \leq p$ such that: 
        \begin{equation}
	                         x_{ij} = \widehat{\phi_{ij}} \left( y_1, y_2 \ldots y_t, N_i \right) = \widehat{\phi_{ij}} \left( y_{m}, N_i \right)   
           \end{equation}    
        for that unique $m$, such that $G_{m,ij} \neq 0$. Here, $N(i)$ is the set of side information packets user $i$ has access to, listed by the directed out-neighborhood of vertex $i$ in the graph $G_d$.  If $j \in N_i$, then
         $\left[x_{j1}~x_{j2} \ldots x_{jp}\right]$ is present as side information at user $i$.
        \end{enumerate}  \hfill $\lozenge$.
	\end{definition}
 
    The optimum vector XOR coloring (the optimum ratio $\frac{t}{p}$) is given by fractional coloring of the complement of the undirected side information graph denoted by $\chi_f(\bar{G})$. The achievable scheme was first outlined in \cite{blasiak2010index}. We briefly mention it here. Fractional coloring number $\chi_{f}\left(\bar{G}\right)$ is the optimum of the following linear program:
   
         \begin{align}\label{Prob:fracchrom}
             \min & \sum \limits_{C \in {\cal MC}(G)} q(C) \nonumber \\
              \mathrm{subject~ to~} & \sum \limits_{C:v \in C} q(C) \geq 1,~ \forall v \in V \nonumber \\
              \hfill & q(C) \geq 0,~\forall C \in {\cal MC}(G),~ q(C) \in \mathbb{R}  
         \end{align}
     where ${\cal MC}(G)$ is the set of all maximal cliques in $G$. Real nonnegative weight $q(C)$ is assigned to every maximal clique $C$ in $G$ such that every vertex is covered by a weight of at least $1$ through only the maximal cliques that contain it. The minimum possible total weight is called the fractional chromatic number. Since the program has integral coefficients, the optimum $q(C)$ is always rational for all $C$. Let $p$ denote the greatest common divisor of all denominators of $q(C)$ in the optimum solution. Then, setting $q'(C)=p q(C) ~ \forall C $, the covering constraint becomes $\sum \limits_{C:v \in C} q'(C) \geq p$. Let the total number of colors be $t= \sum \limits_{C \in {\cal MC}(G)} q'(C)$. This means every maximal clique $C$ is assigned $q'(C)$ integral colors such that every vertex gets at least $p$ colors. And the total number of colors is $t$. Now it is possible to reallocate the colors from the set of maximal cliques to the set of all cliques in such a way that every vertex is colored exactly $p$ times and the total number of colors (objective function value) $t$ remains the same. We will not go into how the reallocation is done. We refer to \cite{mehrotra2007branch} for details of reallocation of colors from the set of maximal cliques to the set of cliques while maintaining the total number of colors. Therefore, after solving the linear program, every clique (even non maximal ones) is assigned a nonnegative integral weight such that every vertex is covered exactly by $p$ colors and the total number of colors used is $t$ and $\chi_f \left(\bar{G}\right)=t/p$ is the smallest possible value.
     
      Now, we specify the vector XOR coloring achievable scheme corresponding to the fractional coloring solution, with the help of an example. Let us say that the number of colors assigned to vertex $1$ is $p$ and out of that $p'<p$ colors have been assigned to a clique $C=\{1,2,3\}$. Let us number all the colors from $1$ to $t$. Let color $m$ be assigned to $C$. Let $i,j$ and $k$ represent the positions of color $m$ in the different sorted orderings of $p$ colors assigned to vertices $1,2$ and $3$ respectively. Then, the $m$-th transmission is an XOR of sub-packets $x_{1i},x_{2j}$ and $x_{3k}$, i.e. $y_m=x_{1i}+x_{2j}+x_{3k}$ as in (\ref{Eqn:Colorenc}). Since $C$ is a clique in the side information graph $G$, each of the users $1,2$ and $3$ decodes exactly one sub-packet using the side information available. Clearly, every sub-packet participates in exactly one transmission corresponding to a clique. Therefore, fractional chromatic number corresponds to the broadcast rate of the optimal vector XOR coloring scheme.
      
      Clearly, (\ref{Prob:fracchrom}) is a linear program with exponentially many variables (one for each maximal clique) for a general graph. If we consider an undirected $k$-partite graph $G^u$ with the complete bipartite structure, we show that using the machinery of section \ref{categorysec}, problem (\ref{Prob:fracchrom}) is equivalent to the following linear program which corresponds to a fractional multi-coloring problem on the complement of the category graph ($\bar{{\cal G}}$):
      
        \begin{align}\label{Prob:fracmulti}
             \min & \sum \limits_{S \in {\cal MC}({\cal G})} c(S) \nonumber \\
              \mathrm{subject~ to~} & \sum \limits_{S:v \in S} c(S) \geq w(v),~ \forall v \in {\cal V} \nonumber \\
              \hfill & c(S) \geq 0,~\forall S \in {\cal MC}({\cal G}),~ c(S) \in \mathbb{R}  
         \end{align}
     Here, ${\cal G}\left({\cal V},{\cal E}\right)$ is the category graph of $G$. ${\cal MC}\left({\cal G}\right)$ is the set of maximal cliques in ${\cal G}$. As before, $S$ represents a maximal clique in the category graph and $c(S)$ represents the amount of "fractional colors" assigned exclusively to the maximal clique $S$. $w(v)$ is the weight of the category vertex $v$ or the number of vertices in the graph $G$ that belongs to category $v$.  We have the following theorem:
     
     \begin{theorem}
          Fractional multi-coloring of $\bar{{\cal G}}$ (Problem \ref{Prob:fracmulti}) is equivalent to the fractional coloring of $\bar{G}$ (Problem \ref{Prob:fracchrom}).
     \end{theorem}
      \begin{proof}
            Let $t(q)$ and $t(c)$ represent the objective functions corresponding to the respective feasible assignments $q(.)$ and $c(.)$ for problems (\ref{Prob:fracchrom}) and (\ref{Prob:fracmulti}) respectively. First, we show that if $\exists q(\cdot)$ satisfying (\ref{Prob:fracchrom}), then $\exists c(\cdot)$ satisfying (\ref{Prob:fracmulti}) such that $t(c) = t(q)$. From the properties of the category graph ${\cal G}$ in section \ref{categorysec}, we have the following observation: Any maximal clique $C=\{i_1,i_2,i_3 \ldots i_m\} \in G$ is such that it corresponds to a unique maximal clique $h(C)=S=\{v_1,v_2 \ldots v_m\} \in {\cal G}$ such that $i_j$ is in category $v_j$. This is because vertices belonging to a category are disconnected and $G$ has the complete bipartite structure. Now, we can set $c(S)= \sum \limits_{C:h(C)=S} q(C)$. Since a category $v$ contains $w(v)$ vertices, the assignment satisfies all the constraints of Problem (\ref{Prob:fracmulti}). Hence, one direction has been shown.
            
            Now, we show the other direction: if $\exists c(\cdot)$ satisfying (\ref{Prob:fracmulti}), then $\exists q(\cdot)$ satisfying (\ref{Prob:fracchrom}) such that $t(c) = t(q)$.  Let the node $a \in G$ belong to category $v$. Problem (\ref{Prob:fracmulti}) is a linear program with integer coefficients. Therefore, the optimal feasible assignment $c(\cdot)$ is rational. Hence, without loss of generality we assume that $c(S)$ is rational for all $S$. Define $c'(\cdot)= b c(\cdot)$ where $b$ is the least common multiple of all the denominators of the rational numbers $c(S)$. The assignment $c'(\cdot)$ satisfies:
 
  \begin{equation}\label{Eqn:weightfunc}
    \sum \limits_{S: v \in S}c'(S) \geq w(v) b
 \end{equation}
 
 Define the new objective function to be $t(c')=b t(c)$. Define the set of colors used in the new assignment $c'(\cdot)$ to be ${\cal L}= \{1, 2, 3 \ldots.. t(c')\}$. Let ${\cal C}'_{S} \subseteq {\cal L}$ be the unique set of $c'(S)$ numbered colors assigned to the maximal clique $S$. Define ${\cal S}_v=\{S: v \in S\}$. Choose a set of colors $ {\cal C}'_{a} \subset \bigcup \limits_{S\in {\cal S}_v} {\cal C}'_{S}$ such that $\lvert {\cal C}'_a \rvert \geq b$. It is always possible to find a pairwise disjoint sets of colors $\{ {\cal C}'_a \}$ for all the nodes $a$ belonging to the category $v$ such that $\bigcup \limits_{S \in {\cal S}_v} {\cal C}'_{S} = \bigcup \limits_{a \in v} {\cal C}'_a $ due to the constraint (\ref{Eqn:weightfunc}). Consider a maximal clique $C=\{a_1,a_2 \ldots a_m \}$ made of nodes in $G$ from different categories. Define ${\cal Q}_{C}=\bigcap \limits_{a \in C } {\cal C}'_{a}$ for all maximal cliques $C$ in ${\cal G}$.  

Note that a the category set $h(C)=\{v_1,v_2 \ldots v_m\}$ forms a maximal clique in ${\cal G}$ iff $C$ is a maximal clique in $G$. Therefore, by this procedure, any color in ${\cal L}$ is assigned to some ${\cal Q}_{C}$. We will now show that $C_1 \neq C_2 $ implies that ${\cal Q}_{C_1} \bigcap {\cal Q}_{C_2} = \emptyset$. We prove this by contradiction. Let $S_1$ (resp. $S_2$) corresponds to the maximal clique of categories corresponding to nodes in $C_1$ (resp. $C_2$). Consider $x \in {\cal Q}_{C_1} \bigcap {\cal Q}_{C_2}$ and $C_1 \neq C_2$. This implies that there is a color under assignment $c'(\cdot)$ assigned to different maximal cliques $S_1$ and $S_2$, i.e. $x \in  {\cal C}'_{S_1}\bigcap {\cal C}'_{S_2}$. But, this is not possible as no color is assigned to more than one maximal clique $S$. Therefore, ${\cal Q}_{C_1} \bigcap {\cal Q}_{C_2} = \emptyset$. Now, set $q(C)= \frac{\lvert {\cal Q} (C) \rvert}{b}$. Since $\lvert {\cal C}'_a \rvert \geq b ~\forall a$, it is easy to see that $q(\cdot)$ satisfies the constraints of (\ref{Prob:fracchrom}) and $t(q) =t(c)$.          
 
    \end{proof}
    
    Hence, the optimal Vector XOR coloring solution can be obtained by the linear program in (\ref{Prob:fracmulti}) which is a relaxation of (\ref{Eqn:integcol}). The number of variables is atmost $2^{2k^2}$ and each variable can take a value of at most $n$. Hence, the running complexity of this LP is $2^{O\left( k^2 \right)} n$. 
    
    \section{Simulation Results}
           In this section, we provide some simulation results regarding the performance of the vector XOR coloring algorithm for the ICDH problem with constant number of helpers ($k$).  The number of users for the simulation results is $600$. The number of users in the neighborhood of a particular helper is $\frac{600}{k}$. Every user is connected to at most one helper. We assume that the user requests for files arise from a library of $1400$ files. Every request is a random sample according to a zipf distribution on $1400$ files with the zipf parameter $0.5$. Let the number of files stored per caching helper be $M$. To populate the caches in helpers, random samples are drawn from the zipf distribution. The user requests are again drawn from the same zipf distribution. If a user request incurs a cache hit from the helper nearby then the particular request is replaced by resampling till a cache miss occurs. Note that it may happen that a particular user request may not be cached anywhere. Finally, we have a random set of user requests all of which have a cache miss with respect to a random cache state both determined by the zipf distribution. We apply the linear program in section \ref{sec:AlgvectorXOR} to determine the broadcast rate of the optimal vector XOR coloring solution. 
           
           The results presented are averaged ($20$ times for every realization) over random realizations of user requests and helper cache states. We compare the number of transmissions (or the broadcast rate) of the optimal vector XOR coloring solution, a greedy algorithm that finds a matching over the side information graph induced by the problem and the naive scheme of transmitting $600$ packets. A matching is a set of undirected edges such that every vertex is covered at most once by the edges. Each edge corresponds to an XOR transmission of the corresponding packets requested by the users represented by the vertices participating in the edge. A greedy matching is a low complexity heuristic to produce a good index coding solution. The plot in Fig. \ref{Fig:Comparison} shows the number of transmissions versus cache capacity per helper under different transmission schemes when $k=5,6,8$. We observe that the vector XOR coloring gives a constant gain ( about $2.4$) when the storage capacity per helper is about $450$ files over random cache states and random user requests. This shows that, with constant number of helpers, the vector XOR coloring can reduce the number of transmissions by a constant factor with reasonable storage capacity per helper.                      
           \begin{figure}
             \centering
             \includegraphics[width=11cm]{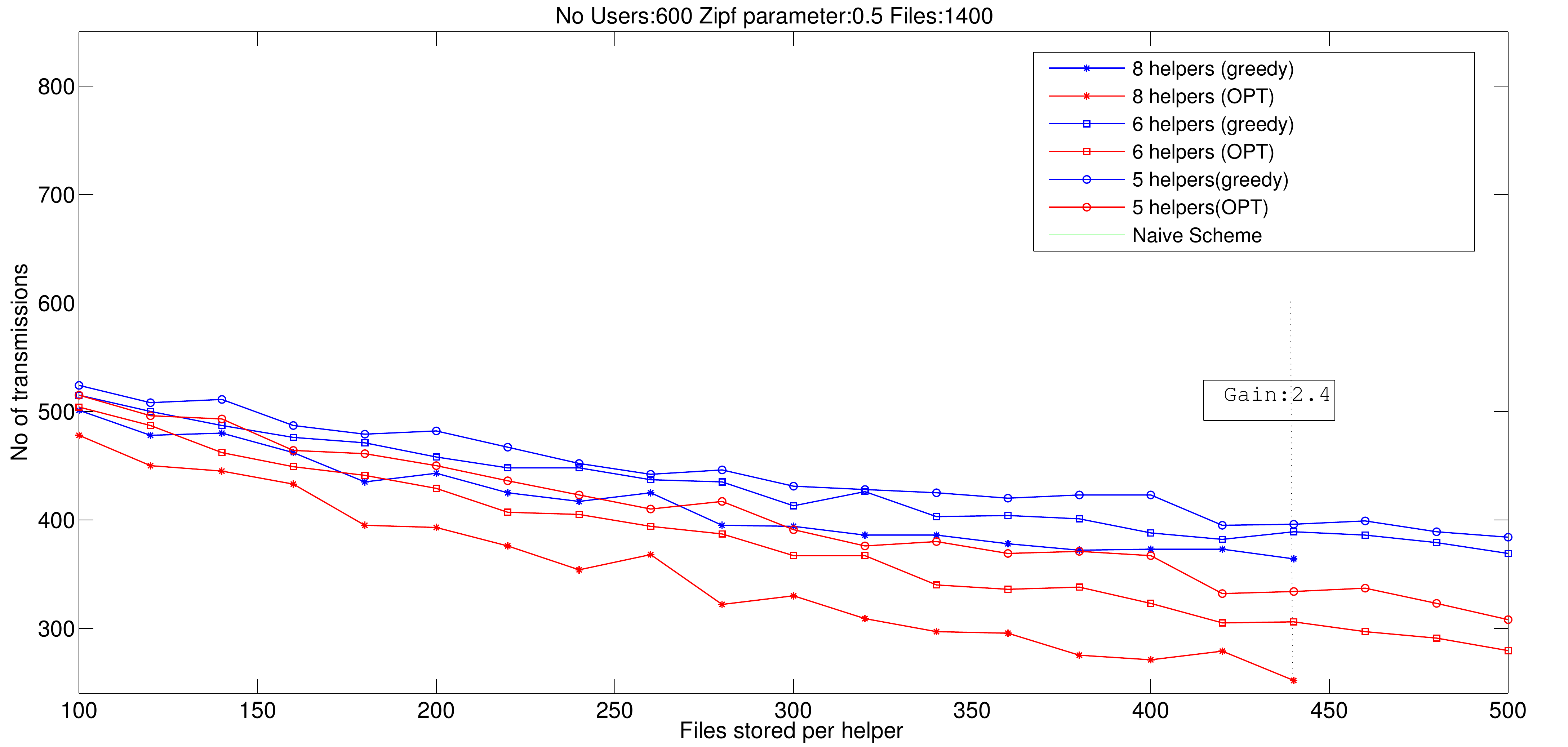}
             \caption{Number of transmissions versus storage capacity per helper for different schemes.}
             \label{Fig:Comparison}
           \end{figure}

	\section{conclusion}
	  We considered a special case of the index coding problem with $k$ caching helpers and $n$ users connected to a base station. When there is no cache hit for any user (file request not found in its neighboring helpers), users receive coded transmission from base station which is cognizant of the cache status of the helpers. Further, the cached files in the neighboring helpers act as side information for users. With interference constraints on helper placement and planar topology, we reduced a general problem, where users are connected to multiple but limited number of helpers, to a canonical form in which every user is connected to just one helper. For the equivalent problem reduced to such a Òcanonical formÓ, the problem of finding the optimal XOR coloring (resp. vector XOR coloring) scheme based on coloring (resp. fractional coloring) of the complement of the side information graph reduces to a problem of multi coloring (resp. fractional multi coloring) a fixed prototypical graph uniquely determined by $k$ only. Using this machinery, we showed that the achievable schemes are computable in time polynomial in $n$ if $k$ is a constant. 
	  
	  For future work, it would be interesting to show a corresponding tractability result with constant number of caching helpers for the scalar linear achievable schemes based on minrank. Further, it would be interesting to reduce the dependence on $k$ in the running time ($2^{O(k^2)} n$) of the fractional multicoloring algorithm for the vector XOR coloring scheme. It would also interesting to see if any low complexity greedy algorithm for $k >3$ can be designed for solving the XOR coloring scheme similar to the one outlined for the special case of $k=3$. 

         \appendix [Proof of Theorem \ref{Thm:XORcol}]
         The algorithm and the proof follow from results in \cite{onn2010nonlinear}. We outline them for the sake of exposition. The integer linear program to be solved is given by (\ref{IntegerProgCol}). $\mathbf{A}$ is a matrix with entries drawn from $\{0,1\}$. Let $\mathbf{c} \in \mathbb{Z}^{f_2(k)}$. Let the domain of integer vectors be $S= \{ \mathbf{c} \in \mathbb{Z}^{f_2(k)}: \mathbf{A}\mathbf{c}=\mathbf{w},~ n \mathbf{1} \geq \mathbf{c} \geq 0 \}$.  Here, we  assume that $\mathbf{c} \leq n\mathbf{1}$ (this is true since the total number of vertices is $n$ in the graph) for ensuring that the domain is bounded. Let the integral null space of $\mathbf{A}$ be defined as: ${\cal L}(\mathbf{A})=\{ \mathbf{c} \in \mathbb{Z}^{f_2(k)} : \mathbf{A}\mathbf{c}=0,~ \mathbf{c} \neq 0\}$.  
	
	\begin{definition}
	       Two vectors $\mathbf{u},\mathbf{v} \in \mathbb{R}^{f_2(k)}$ are said to obey the partial order $\mathbf{u} \sqsubset \mathbf{v}$(or one is \textit {conformal} to the other) when they are both in the same orthant and $\lvert u_i \rvert \leq \lvert v_i \rvert,~\forall i$ (component wise $\mathbf{v}$ dominates $\mathbf{u}$ in absolute value) \cite{onn2010nonlinear}. \hfill $\lozenge$
	\end{definition}
	 
	\begin{definition}
	       Graver Basis of $\mathbf{A}$ is defined to be the set ${\cal G}(\mathbf{A})$ of $\sqsubset$-minimal elements in ${\cal L}(\mathbf{A})$ or the set of conformal minimal elements of the integral null space of $\mathbf{A}$ \cite{onn2010nonlinear}. \hfill $\lozenge$
	\end{definition}
	
	\begin{definition}
	       A finite sum $\mathbf{v} = \sum \mathbf{a}_i$ is a conformal sum  if all summands $\mathbf{a}_i$ lie in the same orthant as $\mathbf{v}$ \cite{onn2010nonlinear}. \hfill $\lozenge$
	\end{definition}
	
	\begin{lemma}
	        Any $\mathbf{c} \in {\cal L} (\mathbf{A}) $ is a conformal sum $\mathbf{c} = \sum \mathbf{g}_i$ of graver basis elements $\mathbf{g}_i \in {\cal G}(\mathbf{A})$ with, possibly, some elements appearing more than once \cite{onn2010nonlinear}. \hfill $\square$
	\end{lemma}
	 Now we state the results from \cite{onn2010nonlinear} with some modification to the case under consideration here.
	
	\begin{lemma}
	       \cite{onn2010nonlinear}  Define an augmentation oracle which, when given any feasible point $\mathbf{a} \in S$ to start with, checks if there exists a $\mathbf{g} \in {\cal G}(\mathbf{A}): \mathbf{1}^T \mathbf{g} >0,~ 0 \leq \mathbf{a}+\mathbf{g} \leq n \mathbf{1} $. If it exists, the oracle returns  $\mathbf{a}+\mathbf{g}$, or otherwise, returns $\mathbf{a}$. Then the augmentation oracle returns the supplied input feasible point if the input is optimal for (\ref{IntegerProgCol}) or otherwise produces a better feasible solution. \hfill $\square$
	\end{lemma} 
	
	Note that one augmentation step takes $O(f_3(k)f_2(k) \log n )$ where $f_3(k)$ is the number of vectors in the graver basis and $\log n$ is due to the encoding complexity of every entry value. We give an algorithm, named $\rm{OPTINTPROG}$, (modified for purposes here from \cite{onn2010nonlinear}) which issues subsequent calls to the above augmentation oracle to improve on the current objective value to reach the optimum. Also, we bound the number of such calls. 
	\begin{lemma}
	         Intialize the algorithm $\rm{OPTINTPROG}$ with the feasible solution where $c_v = w(v)$ (the clique variable, corresponding to a singleton category vertex, is set to the weight of the vertex) and $c_{\bar{s}}=0,~ \forall \bar{s}:\ell(\bar{s}) >1$ ($\ell$ denotes the length of the string). Call the augmentation oracle with this initial feasible solution. If the output returned is different from the feasible solution passed, then the output is used to call the oracle again. This procedure is followed till the optimum is reached. The algorithm $\rm{OPTINTPROG}$ issues at most $2nf_2(k)$ calls to the augmentation oracle. \hfill $\square$
	\end{lemma}
	\begin{proof}
	          The proof bounding the running time of  this algorithm is given in \cite{onn2010nonlinear} for a bit scaling version of the algorithm for a general linear objective $\mathbf{x}^T\mathbf{c}$. Since in our case, $\mathbf{x}^T= \mathbf{1}^T$, the algorithm and the complexity guarantee reduces to what is mentioned in the lemma as bit scaling is not necessary.
	\end{proof} 
	Hence, the time complexity of finding the optimal solution is $O(f_3(k)(f_2(k))^2 n \log n)$. The complexity with respect to $k$ for this algorithm depends on the number of vectors in the graver basis ${\cal G}(\mathbf{A})$. Computing the graver basis may take a long time, but can be done offline since the matrix $\mathbf{A}$ is fixed for a given $k$. We are not able to find general bounds on $f_3(k)$ for any $k$. It is known that the number of elements on the graver basis for any matrix is finite (in our case it depends on $k$). We provide the algorithm from \cite{onn2010nonlinear} to compute the graver basis which can be carried out offline. Let $r(k)$ denote the rank of the matrix $\mathbf{A}$. Let $c_i$ denote the $i$-th element in the vector $\mathbf{c}$.
	
	\begin{enumerate}
	          \item Enumerate all integer null space vectors, $\mathbf{c} \in {\cal L}(\mathbf{A})$, such that $\max \limits_i \lvert c_i \rvert \leq (f_2(k)-r(k)) \Delta(\mathbf{A})$ where $\Delta(\mathbf{A})$ is the largest determinant of a square submatrix of $\mathbf{A}$. 
	         \item Find the $\sqsubseteq$ minimal elements from this and add them to the graver basis. 
	\end{enumerate} 
	This solution is aimed at solving the clique cover problem for very small $k$.

	 \pagenumbering{arabic}
	\bibliographystyle{IEEEtran}
	\bibliography{Indhelp}

	\end{document}

%% file: macros.tex
\long\def\comment#1{}

\newfont{\bbb}{msbm10 scaled 700}

\newfont{\bb}{msbm10 scaled 1100}

\newcommand{\Pc}{{\cal P}}

